\documentclass[a4paper,USenglish]{lipics}

\usepackage[utf8]{inputenc}

\usepackage{hyperref}
\usepackage{enumerate}
\usepackage{amsmath}
\usepackage{amssymb}

\usepackage{listings}
\lstdefinelanguage{pseudo}{
	morekeywords={if, then, else, while, foreach, do, assume, assert, let, return},
	sensitive=false,
	morecomment=[l]{//},
	morecomment=[s]{/*}{*/},
	morestring=[b]",
}
\AtBeginDocument{

}

\newcommand{\loopt}{{\scriptstyle\mathrm{LOOP}}} 
\newcommand{\stemt}{{\scriptstyle\mathrm{STEM}}} 
\newcommand{\abovebelow}[2]{
\left(\begin{smallmatrix}
{#1} \\ {#2}
\end{smallmatrix}\right)
} 

\def\copyrightline{}

\title{Geometric Series as Nontermination Arguments for Linear Lasso Programs}

\author[1]{Jan Leike}
\author[2]{Matthias Heizmann}
\authorrunning{J.\ Leike and M.\ Heizmann}

\affil[1]{The Australian National University \\
  Canberra, Australia \\
  \texttt{jan.leike@anu.edu.au}}
\affil[2]{University of Freiburg \\
  Freiburg, Germany \\
  \texttt{heizmann@informatik.uni-freiburg.de}}

\Copyright{Jan Leike and Matthias Heizmann}


\serieslogo{}
\volumeinfo
  {Billy Editor and Bill Editors}
  {2}
  {Conference title on which this volume is based on}
  {1}
  {1}
  {1}
\EventShortName{}

\begin{document}

\maketitle

\begin{abstract}
We present a new kind of nontermination argument for linear lasso programs, called \emph{geometric nontermination argument}.
A geometric nontermination argument is a finite representation of an infinite execution of the form 
$(\vec{x} + \sum_{i=0}^t \lambda^i \vec{y})_{t \geq 0}$.
The existence of this nontermination argument can be stated as a set of nonlinear algebraic constraints.
We show that every linear loop program that has a bounded infinite execution also has a geometric nontermination argument.
Furthermore, we discuss nonterminating programs that do not have a geometric nontermination argument.
\end{abstract}


\section{Introduction}
\label{sec:introduction}
The problem of automatically proving termination of programs has been extensively studied. 
For restricted classes of programs there are
methods proving termination~\cite{Braverman06,Tiwari04} and hence nontermination follows from the absence of a termination proof.
For broader classes of programs no complete method for proving termination is known or termination is undecidable.
Methods that address these broader classes of programs only check the existence of a certain kind of termination argument,
e.g.\ a specific kind of ranking function.
The existence of this termination argument proves termination,
however the absence of such a termination argument does not imply nontermination
and hence these termination analyses cannot be used to prove nontermination.

Analyses for nontermination proceed in a similar manner.
They do not check the existence of a general nontermination proof,
instead they check for the existence of a certain kind of nontermination argument,
e.g.\ a recurrence set~\cite{GHMRX08, Rybalchenko10} or
an underapproximation of the program that does not terminate for any input~\cite{CCFNO14}.

In this paper we present a new kind of nontermination argument for linear lasso programs, called \emph{geometric nontermination argument}.
A geometric nontermination argument is a finite representation of an infinite execution that can be denoted as a geometric series.
The existence of a geometric nontermination argument can be encoded by a set of nonlinear constraints.
Over the reals these constraints are decidable.
The advantage of our nontermination arguments lies in their simplicity.
In contrast to recurrence sets~\cite{GHMRX08, Rybalchenko10},
the constraints that state the existence of our geometric nontermination arguments do not contain quantifier alternation and contain only a small number of nonlinear terms.
Unlike~\cite{CCFNO14} we do not need a safety checker to compute nontermination arguments.


\section{Preliminaries}
\label{sec:preliminaries}
We consider the following class of programs whose states are real-valued vectors.

\begin{definition}[Linear lasso program]\label{def:linear-lassos}
A \emph{(conjunctive) linear lasso program} $P = (\stemt, \loopt)$ consists of two binary relations $\stemt$ and $\loopt$,
that are each defined by a formula whose free variables are $\vec{x}$ and $\vec{x}'$ and that have the form
$A\abovebelow{\vec{x}}{\vec{x}'} \leq \vec{b}$
for some matrix $A \in \mathbb{R}^{n \times m}$
and some vector $\vec{b} \in \mathbb{R}^{m}$.
We call a linear lasso program \emph{linear loop program} if the formula that defines the relation $\stemt$ is equivalent to $true$.
\end{definition}

\begin{definition}[Infinite execution]\label{def:nontermination}
An infinite sequence of states $(\vec{x}_t)_{t \geq 0}$ is an \emph{infinite execution} of the linear lasso program $P = (\stemt,\loopt)$ iff
$(\vec{x}_0, \vec{x}_1) \in \stemt$ and $(\vec{x}_t, \vec{x}_{t+1}) \in \loopt$ for all $t \geq 1$.
\end{definition}


\section{Geometric Nontermination Arguments}
\label{sec:geometric-nontermination-arguments}

\begin{definition}\label{def:nt-argument}
Let $P = (\stemt, \loopt)$ be a linear lasso program such that $\loopt$ is defined by the formula $A \abovebelow{\vec{x}}{\vec{x}'} \leq \vec{b}$.
The tuple $N = (\vec{x}_0, \vec{x}_1, \vec{y}, \lambda)$
is called a \emph{geometric nontermination argument}
for $P$ iff the following properties hold.
\begin{center}
\begin{minipage}{8cm}
\begin{itemize}
\itemsep0.5mm
\setlength{\itemindent}{4em}
\item[(domain)] $\vec{x}_0, \vec{x}_1, \vec{y} \in \mathbb{R}^n$,
$\lambda \in \mathbb{R}$ and $\lambda > 0$.
\item[(init)]
  $(\vec{x}_0, \vec{x}_1) \in \stemt$
\item[(point)]
  $A \abovebelow{\vec{x}_1}{\vec{x}_1 + \vec{y}} \leq \vec{b}$
\item[(ray)]
  $A \abovebelow{\vec{y}}{\lambda \vec{y}} \leq \vec{0}$
\end{itemize}
\end{minipage}
\end{center}
\end{definition}
The constraints (init), (point), and (ray) given in Definition~\ref{def:nt-argument} are (quantifier free) nonlinear algebraic constraints,
the existence of a solution is decidable~\cite{Tarski51},
and hence the existence of a geometric nontermination argument is decidable.
We can check the existence of a geometric nontermination argument by passing the constraints of Definition~\ref{def:nt-argument} to an SMT solver for nonlinear real arithmetic~\cite{JM12}.
If a satisfying assignment is found,
this constitutes a nontermination proof in form of an infinite execution according to the following theorem.

\begin{theorem}[Soundness]\label{thm:soundness}
If the conjunctive linear lasso program $P = (\stemt, \loopt)$ has a geometric nontermination argument $N = (\vec{x}_0, \vec{x}_1, \vec{y}, \lambda)$ then $P$ has the following infinite execution.
$$
\vec{x}_0, \;
\vec{x}_1, \;
\vec{x}_1 + \vec{y}, \;
\vec{x}_1 + (1 + \lambda) \vec{y}, \;
\vec{x}_1 + (1 + \lambda + \lambda^2) \vec{y}, \;
\ldots
$$
\end{theorem}
\begin{proof}
Define $\vec{z}_0 := \vec{x}_0$ and
$
\vec{z}_t := \vec{x}_1 + \sum_{i=0}^t \lambda^i \vec{y}.
$
Then $(\vec{z}_t)_{t \geq 0}$ is an infinite execution of $P$:
by (init), $(\vec{z}_0, \vec{z}_1) = (\vec{x}_0, \vec{x}_1) \in \stemt$ and
\begin{align*}
A \abovebelow{\vec{z}_t}{\vec{z}_{t+1}}
= A \abovebelow{\vec{x}_1 + \sum_{i=0}^t \lambda^i \vec{y}}{\vec{x}_1 + \sum_{i=0}^{t+1} \lambda^i \vec{y}}
= A \abovebelow{\vec{x}_1}{\vec{x}_1 + \vec{y}} + \sum_{i=0}^t \lambda^i A \abovebelow{\vec{y}}{\lambda \vec{y}}
\leq \vec{b} + \sum_{i=0}^t \lambda^i \vec{0}
= \vec{b},
\end{align*}
by (point) and (ray).
\end{proof}

\begin{example}\label{ex:running}
Consider the linear loop program $P = (true, \loopt)$ depicted as pseudocode on the left and whose relation $\loopt(a, b, a', b')$ is defined by the formula depicted on the right.
\begin{center}
\begin{minipage}{30mm}
\begin{lstlisting}
while ($a \geq 7$):
    $a$ := $b$;
    $b$ := $a + 1$;
\end{lstlisting}
\end{minipage}
\hspace{10mm}
\scriptsize
$
\left(\begin{matrix}
-1 & 0 & 0 & 0 \\
0 & -1 & 1 & 0 \\
0 & 1 & -1 & 0 \\
-1 & 0 & 0 & 1 \\
1 & 0 & 0 & -1 
\end{matrix}\right)
\left(\begin{matrix}
a \\ b \\ a' \\ b'
\end{matrix}\right)
\leq
\left(\begin{matrix}
7 \\ 0 \\ 0 \\ 1 \\ 1
\end{matrix}\right)
$
\end{center}
Note that in this example,
the relation $\loopt$ is defined by an affine-linear transformation
and a guard $a \geq 7$.
In general,
linear lasso programs are defined with linear constraints,
which also allow nondeterministic updates of variables.

For $x_0 = \abovebelow{7}{8}$, $x_1 = \abovebelow{7}{8}$, $y = \abovebelow{1}{1}$ and $\lambda = 1$,
the tuple $N = (x_0, x_1, y, \lambda)$ is a geometric nontermination argument and the following sequence of states is an infinite execution of $P$.
$$\abovebelow{7}{8}, \abovebelow{7}{8}, \abovebelow{8}{9}, \abovebelow{9}{10}, \abovebelow{10}{11}, \dots$$
\end{example}



We are able to decide the existence of a geometric nontermination argument,
however we are not able to decide the existence of an infinite execution
because there are programs that have an infinite execution but no geometric nontermination argument as the following example illustrates.

\begin{example}\label{ex:difficult}
The following linear lasso program has an infinite execution, e.g.\ $\abovebelow{2^t}{3^t}_{t \geq 0}$, but it
does not have a geometric nontermination argument.
\begin{center}
\begin{minipage}{41mm}
\begin{lstlisting}
while ($a \geq 1 \;\land\; b \geq 1$):
    $a$ := $2 \cdot a$;
    $b$ := $3 \cdot b$;
\end{lstlisting}
\end{minipage}
\end{center}
\end{example}


\section{Bounded Infinite Executions}
\label{sec:bounded-infinite-executions}
In this section we show that we can always prove nontermination of linear loop programs if there is a bounded infinite execution.

Let $|\cdot|: \mathbb{R}^n\rightarrow \mathbb{R}$ denote some norm.
We call an infinite execution $(\vec{x}_t)_{t \geq 0}$ \emph{bounded} iff
there is a real number $d \in \mathbb{R}$ such that for each state its norm in bounded by $d$,
i.e.\ $|\vec{x}_t|\leq d$ for all $t$.



\begin{lemma}[Fixed Point]\label{lem:fixed-point}
Let $P = (true, \loopt)$ be a linear loop program. 
The loop $P$ has a bounded infinite execution
if and only if
there is a fixed point $\vec{x}^\ast \in \mathbb{R}^n$ such that
$(\vec{x}^\ast, \vec{x}^\ast) \in \loopt$.
\end{lemma}
\begin{proof}
If there is a fixed point $\vec{x}^\ast$, then the loop has the infinite bounded
execution $\vec{x}^\ast, \vec{x}^\ast, \ldots$.
Conversely, let $(\vec{x}_t)_{t \geq 0}$ be an infinite bounded execution.
Boundedness implies that there is an $d \in \mathbb{R}$ such that $|\vec{x}_t| \leq d$ for all $t$.
Consider the sequence $\vec{z}_k := \frac{1}{k} \sum_{t=1}^k \vec{x}_t$.
\begin{align*}
| \vec{z}_k - \vec{z}_{k+1} |
&= \left| \frac{1}{k} \sum_{t=1}^k \vec{x}_t - \frac{1}{k+1} \sum_{t=1}^{k+1} \vec{x}_t \right|
 = \frac{1}{k(k+1)} \left| (k+1) \sum_{t=1}^k \vec{x}_t -  k \sum_{t=1}^{k+1} \vec{x}_t \right| \\
&= \frac{1}{k(k+1)} \left| \sum_{t=1}^k \vec{x}_t - k \vec{x}_{k+1} \right|
\leq \frac{1}{k(k+1)} \left( \sum_{t=1}^k |\vec{x}_t| + k |\vec{x}_{k+1}| \right) \\
&\leq \frac{1}{k(k+1)} (k\cdot d + k\cdot d)
= \frac{2d}{k+1} \longrightarrow 0 \text{ as } k \to \infty.
 \end{align*}
Hence the sequence $(\vec{z}_k)_{k \geq 1}$ is a Cauchy sequence
and thus converges to some $\vec{z}^\ast \in \mathbb{R}^n$.
We will show that $\vec{z}^\ast$ is the desired fixed point.

For all $t$, the polyhedron $Q := \{ \abovebelow{\vec{x}}{\vec{x}'} \mid A \abovebelow{\vec{x}}{\vec{x}'} \leq b \}$ contains $\abovebelow{\vec{x}_t}{\vec{x}_{t+1}}$ and is convex.
Therefore for all $k \geq 1$,
$$
\frac{1}{k} \sum_{t=1}^k \abovebelow{\vec{x}_t}{\vec{x}_{t+1}} \in Q.
$$
Together with
$$
\abovebelow{\vec{z}_k}{\frac{k+1}{k} \vec{z}_{k+1}}
= \frac{1}{k} \abovebelow{\vec{0}}{\vec{x}_1} + \frac{1}{k} \sum_{t=1}^k \abovebelow{\vec{x}_t}{\vec{x}_{t+1}}
$$
we infer
$$
\left( \abovebelow{\vec{z}_k}{\frac{k+1}{k} \vec{z}_{k+1}} - \frac{1}{k} \abovebelow{\vec{0}}{\vec{x}_1} \right) \in Q,
$$
and since $Q$ is closed we have
$$
\abovebelow{\vec{z}^\ast}{\vec{z}^\ast} =
\lim_{k \to \infty}
  \left( \abovebelow{\vec{z}_k}{\frac{k+1}{k} \vec{z}_{k+1}} - \frac{1}{k} \abovebelow{\vec{0}}{\vec{x}_1} \right) \in Q.
\qedhere
$$
\end{proof}

Because fixed points give rise to trivial geometric nontermination arguments,
we can derive a criterion for the existence of geometric nontermination arguments from Lemma~\ref{lem:fixed-point}.

\begin{corollary}
If the linear loop program $P = (true, \loopt)$ has a bounded infinite execution,
then it has a geometric nontermination argument.
\end{corollary}
\begin{proof}
By Lemma \ref{lem:fixed-point} there is a fixed point $\vec{x}^\ast$ such that
$(\vec{x}^\ast, \vec{x}^\ast) \in \loopt$.
We choose
$\vec{x}_1 = \vec{x}^\ast$,
$\vec{y} = \vec{0}$, and
$\lambda = 1$,
which satisfies (point) and (ray) and thus is a geometric nontermination argument for $P$.
\end{proof}


\begin{example}\label{ex:strict}
Note that according to our definition of a linear lasso program, the relation $\loopt$ is a topologically closed set. 
If we allowed the formula defining $\loopt$ to also contain strict equalities, Lemma~\ref{lem:fixed-point} no longer holds:
the following program is nonterminating and has a bounded infinite execution,
but it does not have a fixed point.
However, the topological closure of the relation $\loopt$ contains the fixed point $x^* = 0$.
\begin{center}
\begin{minipage}{29mm}
\begin{lstlisting}
while ($x > 0$):
    $x$ := $\tfrac{1}{2} \cdot x$;
\end{lstlisting}
\end{minipage}
\vspace{-4mm}
\end{center}
\end{example}

\section{Discussion}
\label{sec:discussion}

\subsection{Recurrence Sets}
\label{ssec:recurrence-sets}

Nontermination arguments related to ours are \emph{recurrence sets}~\cite{GHMRX08, Rybalchenko10}.
A recurrence set $S$ is a set of states such that
\begin{itemize}
 \item at least one state of $S$ is in the range of $\stemt$,  i.e.\ 
 $$\exists \vec{x}, \vec{x}'. (\vec{x},\vec{x}')\in\loopt\land \vec{x'}\in S, \text{ and}$$
 \item for each state in $S$ there is at least one $\loopt$-successor that is in $S$, i.e., 
 $$\forall \vec{x}. \vec{x}\in S \rightarrow \exists \vec{x}' (\vec{x},\vec{x}')\in\loopt .$$
\end{itemize}
If we restrict the form of $S$ to a convex polyhedron,
we can encode its existence using algebraic constraints~\cite{GHMRX08, Rybalchenko10}
and hence decide the existence of such a recurrence set.
However these algebraic constraints are not easy to solve;
they contain nonlinear arithmetic and quantifier alternation that cannot be eliminated with Farkas lemma if the program is nondeterministic.
In contrast to these constraints, our constraints (init), (point), and (ray) contain at most one nonlinear term for each dimension of the state space.

However, recurrence sets are more general nontermination arguments than geometric nontermination arguments as shown by the following lemma.

\begin{lemma}
Let $P = (\stemt, \loopt)$ be a linear lasso program and 
$N = (\vec{x}_0, \vec{x}_1, \vec{y}, \lambda)$ be a geometric nontermination argument for $P$.
The following set $S$ is a recurrence set for $P$.
$$
S = \Big\{ \vec{x}_1 + \sum_{i = 0}^t \lambda^i \vec{y} \mid t \in \mathbb{N} \Big\}
$$
\end{lemma}
\begin{proof}
The state $\vec{x}_1$ is in the range of $\stemt$ by (init).
Furthermore, for $\vec{x}_1 + \sum_{i = 0}^t \lambda^i \vec{y} \in S$,
$\vec{x}_1 + \sum_{i = 0}^{t+1} \lambda^i \vec{y} \in S$ and
$(\vec{x}_1 + \sum_{i = 0}^t \lambda^i \vec{y},\; \vec{x}_1 + \sum_{i = 0}^{t+1} \lambda^i \vec{y}) \in \loopt$ according to the proof of \autoref{thm:soundness}.
\end{proof}

Furthermore, for every geometric nontermination argument $N = (\vec{x}_0, \vec{x}_1, \vec{y}, \lambda)$ there exists a recurrence set $S$ that is a polyhedron.
$$
S = \{ \vec{x} \in \mathbb{R}^n \mid
  \vec{y}^T (\vec{x} - \vec{x}_1) \geq 0 \;\land\;
  \bigwedge_{i \in I} \vec{z}_i^T (\vec{x} - \vec{x}_1) = 0 \},
$$
where $(\vec{z}_i)_{i \in I}$ is a span of the vector space orthogonal to $\vec{y}$.
(For $\lambda < 1$ we need to add the additional constraint
$\vec{y}^T (\vec{x} - \vec{x}_1) \leq \vec{y}^T (\vec{x}_1 + \frac{1}{1 - \lambda} \vec{y})$.)
%

\subsection{Integers vs.\ Reals}

A nonterminating program over the reals may terminate over the integers.
If we restrict the states of the linear lasso program to integer-valued vectors,then \autoref{thm:soundness} only holds if we restrict the values for the variables $\vec{x}_0, \vec{x}_1, \vec{y}$, and $\lambda$ in the constraints (init), (point), and (ray) to integers.
Satisfiability of nonlinear arithmetic over the integers is undecidable and we do not know if our constraints fall into a decidable subclass of this problem. However, we may fix the value of $\lambda$ in advance to a finite set of values. If we do so, we do not have completeness (we may not find every geometric nontermination argument) but we obtain linear arithmetic constraints, which can be solved efficiently.


\bibliographystyle{plain}
\bibliography{references}

\end{document}